\documentclass[a4paper]{article}
\usepackage{pdfsync}
\usepackage[latin1]{inputenc}
\usepackage{graphicx}
\usepackage{amssymb,amsfonts,amsmath}
\usepackage{theorem}
\usepackage{hyperref}
\usepackage{graphicx}
\usepackage{multirow}
\usepackage{verbatim}

\usepackage{color}

\newtheorem{theorem}{Theorem}[section]

\newtheorem{lemma}{Lemma}[section]

{\theorembodyfont{\rmfamily}

\newtheorem{remark}{Remark}[section]

}
\newenvironment{proof}[1][Proof]{\noindent\textbf{#1.} }{\newline \hspace*{\textwidth}\hspace*{-0,4cm} \rule{0.5em}{0.5em} \vspace{0,2cm}}

\begin{document}

\title{\textbf{Kinematic relative velocity with respect to stationary observers in Schwarzschild spacetime}}


\author{Vicente J. Bol\'os
\\{\small Dpto. Matem\'aticas para la Econom\'{\i}a y la Empresa, Facultad de Econom\'{\i}a,}\\ {\small Universidad de Valencia. Avda. Tarongers s/n. 46022, Valencia, Spain.}\\ {\small e-mail\textup{: \texttt{vicente.bolos@uv.es}}}}


\maketitle

\begin{abstract}
We study the kinematic relative velocity of general test particles with respect to stationary observers (using spherical coordinates) in Schwarzschild spacetime, obtaining that its modulus does not depend on the observer, unlike Fermi, spectroscopic and astrometric relative velocities. We study some fundamental particular cases, generalizing some results given in other work about stationary and radial free-falling test particles. Moreover, we give a new result about test particles with circular geodesic orbits: the modulus of their kinematic relative velocity with respect to any stationary observer depends only on the radius of the circular orbit, and so, it remains constant.
\end{abstract}



\section{Introduction}

The concept of ``relative velocity'' of a test particle with respect to an observer in general relativity is only well defined when the observer and the test particle are in the same event. Nevertheless, the notion of relative velocity of a distant test particle is fundamental in physics and so, it was revised by the IAU using reference systems adapted to the solar system (see \cite{Soff03,Lind03}). Thereby, some authors have introduced new geometric concepts motivated by the coordinate-dependence of some definitions; for example, scaled Fermi-Walker derivatives let us define geometrically local notions of velocities of a test particle with respect to a congruence of observers (see \cite{Ca06}).
Moreover, four different intrinsic geometric definitions of relative velocity of a distant test particle with respect to a single observer were introduced in \cite{Bol07}. These definitions are strongly associated with the concept of simultaneity: \textit{kinematic} and \textit{Fermi} in the framework of ``spacelike simultaneity'', \textit{spectroscopic} and \textit{astrometric} in the framework of ``lightlike simultaneity''. These four concepts each have full physical sense, and have proved to be useful in the study of properties of particular spacetimes (see \cite{Bol07,KC10,KR11,BK12}).

Following this line, we are going to study the kinematic relative velocity of test particles with respect to stationary observers in Schwarzschild spacetime. This velocity shows a kind of ``Newtonian behavior'' in this spacetime unlike the other three velocities, and some interesting properties about stationary observers hold, as we are going to develop in the present work.

This paper is organized as follows. In Section \ref{sec2}, we establish notation and define the concept of kinematic relative velocity. In Section \ref{sec3} we introduce the Schwarzschild metric in spherical coordinates and their corresponding stationary observers, giving some lemmas that are applied in Section \ref{sec4} to obtain the main result: the expression of the modulus of the kinematic relative velocity of a general test particle with respect to any stationary observer. We also study in this section some fundamental examples that were previously introduced in \cite{Bol07}, generalizing the results obtained in that paper, and we present a new example about circular geodesic orbits. Finally, Section \ref{sec5} gives concluding remarks that lead to a property that extends to general relativity the main result.

\section{Definitions and notation}
\label{sec2}


We work in a Lorentzian spacetime manifold $\left( \mathcal{M},g\right) $, with $c=1$ and using the `mostly plus' signature convention $(-,+,+,+)$. We suppose that $\mathcal{M}$ is a convex normal neighborhood; thus, given two events $p$ and $q$ in $\mathcal{M}$, there exists a unique geodesic joining them (results on the existence of convex normal neighborhoods in semi-Riemannian manifolds are given in \cite{ON83}, pp. 129--131; see Remark \ref{rem1} for a discussion about working in a non convex normal neighborhood). The parallel transport from $q$ to $p$ along this geodesic is denoted by $\tau _{qp}$. If $\beta :I\rightarrow \mathcal{M}$ is a curve with $I\subseteq \mathbb{R}$ a real interval, we identify $\beta $ with the image $\beta I$ (that is a subset in $\mathcal{M}$), in order to simplify the notation. Vector fields are denoted by uppercase letters and vectors (defined at a single point) are denoted by lowercase letters. Moreover, if $x$ is a spacelike vector, then $\Vert x\Vert :=g\left( x,x\right) ^{1/2}$ is the modulus of $x$. If $X$ is a vector field, $X_p$ denotes the unique vector of $X$ in $T_p\mathcal{M}$.

In general, we say that a timelike world line $\beta $ is an \textit{observer} (or a \textit{test particle}). Nevertheless, we say that a future-pointing timelike unit vector $u$ in $T_{p}\mathcal{M}$ is an \textit{observer at $p$}, identifying the observer with its 4-velocity.

The \textit{Landau submanifold} $L_{p,u}$ (also called \textit{Fermi surface}) is given by all the geodesics starting from $p$ and orthogonal to $u$ (see \cite{Ferm22,Bol02,KR11}).

\subsection{Kinematic relative velocity}
\label{sec2.1}


Throughout the paper, we consider an observer $\beta $ and a test particle $\beta '$ (parameterized by their proper times) with 4-velocities $U$ and $U'$ respectively.
Let $u:=U_p$ be the 4-velocity of $\beta $ at an event $p$ and let $q$ be the event of $\beta '$ such that there exists a spacelike geodesic $\psi $ orthogonal to $u$ joining $p$ and $q$ (see Figure \ref{diagram}). Note that since we work in a convex normal neighborhood, this event is unique and it is given by $q:=L_{p,u}\cap \beta '$. We denote $u':=U'_q$ in order to simplify the notation.

\begin{figure}[tbp]
\begin{center}
\includegraphics[width=0.35\textwidth]{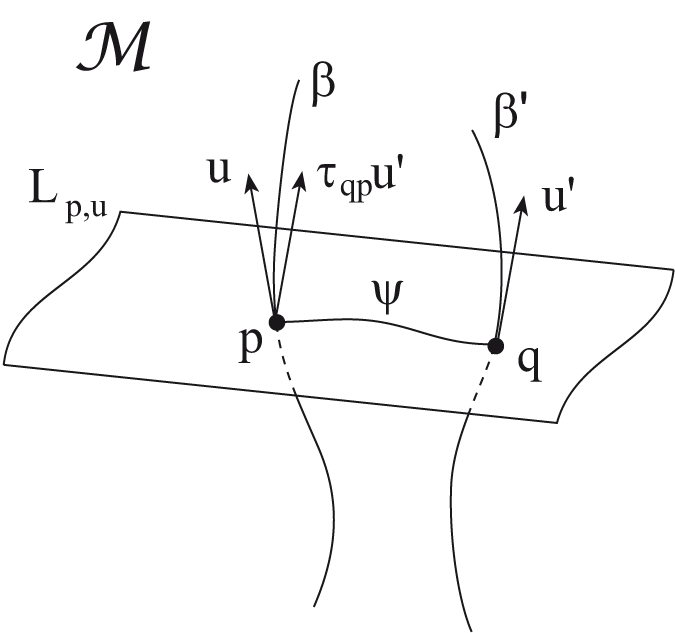}
\end{center}
\caption{Scheme of the elements involved in the definition of the kinematic relative velocity of $\beta '$ (test particle) with respect to $\beta $ (observer). The curve $\psi $ is a spacelike geodesic orthogonal to the 4-velocity of $\beta $ at $p$, denoted by $u$. The vector $u'$ is the 4-velocity of $\beta '$ at $q$.} \label{diagram}
\end{figure}

The \textit{kinematic relative velocity of $u'$ with respect to $u$} is the vector
\begin{equation}
\label{vkin}
v_{\mathrm{kin}}:=\frac{1}{-g\left( \tau _{qp}u',u\right) }\tau _{qp}u'-u.
\end{equation}
In the case $p=q$, this definition coincides with the usual concept of relative velocity
\begin{equation}
\label{vusual}
v=\frac{1}{-g\left( u',u\right) }u'-u,
\end{equation}
which is only defined when $u$ and $u'$ are in the same tangent space.

Note that $v_{\mathrm{kin}}$ is spacelike and orthogonal to $u$, and the square of its modulus is given by
\begin{equation}
\label{vkinmod}
\| v_{\mathrm{kin}}\| ^2=g\left( v_{\mathrm{kin}},v_{\mathrm{kin}}\right) =1-\frac{1}{g\left( \tau _{qp}u',u\right) ^2}=1-\frac{1}{g\left( u',\tau _{pq}u\right) ^2},
\end{equation}
since parallel transport conserves the metric.
Varying $p$ along $\beta $, we construct the vector field $V_{\mathrm{kin}}$ defined on $\beta $, representing the \textit{kinematic relative velocity of $\beta '$ with respect to $\beta $} (see \cite{Bol07,Bol12}).
Throughout the paper we are going to denote $v_{\mathrm{kin}}:=V_{\mathrm{kin}\, p}$ as we have already done in this section.

\begin{remark}
\label{rem1}
The concept of kinematic relative velocity can be extended to non convex normal neighborhoods. If $p$ is an event of the observer $\beta $ with 4-velocity $u$, $q$ is an event of the test particle $\beta '$ with 4-velocity $u'$, and $\psi $ is a spacelike geodesic segment orthogonal to $u$ joining $p$ and $q$, expression \eqref{vkin} defines a kinematic relative velocity of $u'$ with respect to $u$, where $\tau _{qp}$ represents the parallel transport from $q$ to $p$ along $\psi $; in this case, there exists a unique $v_{\mathrm{kin}}$ associated to the set $\left\{ p,q,\psi \right\} $. Working in a convex normal neighborhood implies that given $p$, there exists a unique set $\left\{ p,q,\psi \right\} $ satisfying the above conditions, and so there is a unique $v_{\mathrm{kin}}$ associated to this $p$. But if we work in a non convex normal neighborhood, we could find different sets $\left\{ p,q,\psi '\right\} $, $\left\{ p,q',\psi ''\right\} $, $\ldots $ satisfying the above conditions, and hence, there would be a different $v_{\mathrm{kin}}$ for each set. This extension can be also done for the other concepts of relative velocity and, for example, in the framework of lightlike simultaneity, if there is gravitational lensing then each image of the observed object has a different spectroscopic and astrometric relative velocity.
\end{remark}

\section{Stationary observers in Schwarzschild spacetime}
\label{sec3}

The Schwarzschild metric in spherical coordinates $\left\{ t,r,\theta ,\varphi \right\} $ is given by the line element
\begin{equation}
\label{eq:sch}
ds^2 = -\left( 1-\frac{2m}{r}\right) \mathrm{d}t^2+\left( 1-\frac{2m}{r}\right) ^{-1}\mathrm{d}r^2+r^2\left( \mathrm{d}\theta ^2+\sin ^2\theta \mathrm{d}\varphi ^2 \right) ,
\end{equation}
where the parameter $m$ is interpreted as the mass of the gravitating object, $r>2m$ is the radial coordinate, and $0<\theta <\pi$. From now on we are going to suppose that the coordinates hold these restrictions.

In the framework of this coordinate system, a \textit{stationary observer} is an observer with constant spatial coordinates and its 4-velocity is given by
\begin{equation}
\label{eq:4velsta}
U=\left( 1-\frac{2m}{r}\right) ^{-1/2}\frac{\partial }{\partial t}=\left( \left( 1-\frac{2m}{r}\right) ^{-1/2},0,0,0\right) .
\end{equation}
Note that stationary observers are not geodesic, but they are useful in the description and interpretation of the Schwarzschild spacetime because the vector field $\frac{\partial }{\partial t}$ is Killing. Moreover, we can consider \textit{asymptotic exterior observers} making $r\rightarrow +\infty $ from stationary observers.

Next, we are going to give some lemmas that will be applied in the next section.

\begin{lemma}
\label{lemma1}
Let $p$ be an event of a stationary observer, and let $u=U_p$ be its 4-velocity at $p$, see \eqref{eq:4velsta}. Then, the Landau submanifold $L_{p,u}$ is contained in the hypersurface $t=t_0$, where $t_0$ is the coordinate time of $p$.
\end{lemma}
\begin{proof}
The Landau submanifold is formed by all the spacelike geodesics starting from $p$ and orthogonal to $u$ at $p$; let $\psi $ be one of these geodesics (affinely parameterized) and let $S$ be its tangent vector field. Since the vector field $\frac{\partial }{\partial t}$ is Killing and taking into account \eqref{eq:4velsta}, we have
\begin{equation}
\label{eq:gvpt0}
g\left( S,\frac{\partial }{\partial t}\right) =g\left( S_p,\left. \frac{\partial }{\partial t}\right| _p\right) =\left( 1-\frac{2m}{r_0}\right) ^{1/2} g\left( S_p,u\right) =0,
\end{equation}
where $r_0$ is the radial coordinate of $p$, and $S_p$ is orthogonal to $u$. Hence, from \eqref{eq:gvpt0} and taking into account that the metric \eqref{eq:sch} is diagonal, we have
that $S^t = 0$, and so $\psi ^t$ is constant, concluding that $\psi $ is in the hypersurface $t=t_0$.
\end{proof}


\begin{lemma}
\label{lemma2}
The vector field $U$ given by \eqref{eq:4velsta} is parallel transported along curves in hypersurfaces of the form $t=\mathrm{constant}$, i.e. curves with constant time component.
\end{lemma}

\begin{proof}
Let $\varphi $ be a curve in a hypersurface of the form $t=\mathrm{constant}$ and let $S$ be its tangent vector field. Hence $S^t=0$, and taking into account the Christoffel symbols of the metric, there is only one nontrivial equation for the parallel transport of $U$ along $\varphi $:
\begin{equation}
\label{eq:l2pt}
\frac{\mathrm{d}U^t}{\mathrm{d}r}S^r + \Gamma ^t_{rt} S^r U^t = 0 \quad \Longrightarrow \quad \left( \frac{\mathrm{d}U^t}{\mathrm{d}r} +\frac{m}{r\left( r-2m\right) } U^t\right) S^r = 0.
\end{equation}
Since $U^t=\left( 1-\frac{2m}{r}\right) ^{-1/2}$, equation \eqref{eq:l2pt} holds for any $S^r$.
\end{proof}

\section{Modulus of the kinematic relative velocity with respect to stationary observers}
\label{sec4}

In this section, we are going to work in the Schwarzschild metric using spherical coordinates. Moreover, we are going to consider a stationary observer containing $p=\left( t_0,r_0,\theta _0,\varphi _0\right) $, whose 4-velocity at $p$ is $u=U_p$, where $U$ is the vector field given by \eqref{eq:4velsta}.

\begin{theorem}
\label{theo0}
Given a test particle with 4-velocity $u'$ at $q=\left( t_0,r_1,\theta _1,\varphi _1\right) $, we have
\begin{equation}
\label{vkinmod2}
\| v_{\mathrm{kin}}\| ^2=\| v\| ^2=1-\frac{r_1}{\left( r_1-2m\right) \left( u'^t\right) ^2},
\end{equation}
where $v_{\textrm{kin}}$ is the kinematic relative velocity of $u'$ with respect to the stationary observer $u$, and $v$ is the usual relative velocity of $u'$ with respect to $U_q$.
\end{theorem}
\begin{proof}
Using Lemmas \ref{lemma1} and \ref{lemma2}, we have that $\tau _{pq}u=U_q$. Then, by \eqref{vusual}, \eqref{vkinmod} and \eqref{eq:4velsta} the result holds.
\end{proof}

Since \eqref{vkinmod2} does not depend on $u$, the kinematic relative velocity shows a kind of ``Newtonian behavior'' when it is measured by stationary observers, unlike the other three relative velocities (Fermi, spectroscopic and astrometric), that do not have this behavior in general, (see Figure \ref{velschw}).
Moreover, making $r_0\rightarrow +\infty $ we obtain that Theorem \ref{theo0} also holds for asymptotic exterior observers.

\begin{remark}
Schwarzschild spacetime is not a convex normal neighborhood and then we have to take into account Remark \ref{rem1}. Following the notation of that remark, Lemma \ref{lemma1} assures that given $p\in \beta $ there exists a unique $q\in \beta '$ such that $q\in L_{p,u}$ (because $p$ and $q$ must have the same coordinate time), and viceversa; but there could exist different spacelike geodesics orthogonal to $u$ joining $p$ and $q$, and consequently, different kinematic relative velocities of $u'$ with respect to $u$. Nevertheless, we can conclude from Theorem \ref{theo0} that all of these velocities have the same modulus.
\end{remark}

Next, we are going to study some fundamental particular cases.

\subsection{Stationary test particles}
\label{sec:stp}

Taking into account Theorem \ref{theo0} and \eqref{eq:4velsta}, the kinematic relative velocity of a stationary observer (in the role of a test particle) with respect to any other stationary observer is zero (in \cite{Bol07} it is proved in the particular case of stationary observers aligned with the origin, i.e. with the same $\theta $ and $\varphi $ coordinates). In fact, it can be proved that the Fermi and astrometric relative velocities are also zero, because they are based on changes of relative position. On the other hand, the spectroscopic relative velocity is not zero and produces the gravitational shift (see \cite{Nar94,Bol05,Bol07}). 

\subsection{Radially inward free-falling test particles}

In \cite{Bol07} it is also studied the case of a radially inward free-falling test particle at $r_1$ with respect to a stationary observer at $r_0\geq r_1$ and aligned with the test particle (i.e. with the same $\theta $ and $\varphi $ coordinates). Without loss of generality we can consider that the test particle has $\theta =\pi /2$ (i.e. it is equatorial) and $\varphi =0$ (because it is radial), and so its 4-velocity at $q=\left(t_0,r_1,\pi /2,0\right) $ is given by
\begin{equation}
u'=\left( \frac{E}{1-\frac{2m}{r_1}},-\sqrt{E^2-\left( 1-\frac{2m}{r_1}\right) },0,0 \right), \label{4velfreefall}
\end{equation}
where $E:=\left( \frac{1-2m/r_{\textrm{ini}}}{1-v_{\textrm{ini}}^2}\right) ^{1/2}$ is a constant of motion, $r_{\textrm{ini}}$ is the
radial coordinate at which the fall begins, and $v_{\textrm{ini}}$ is the initial
velocity (see \cite{Craw02}). It was obtained that the square of the modulus of the kinematic relative velocity with respect to an aligned stationary observer is
\begin{equation}
\label{vkinff}
\left\Vert v_{\mathrm{kin}}\right\Vert ^2=1-\frac{1-\frac{2m}{r_1}}{E^2}.
\end{equation}
From Theorem \ref{theo0} and the expression of $u'^t$ given in \eqref{4velfreefall}, we can generalize this result obtaining that \eqref{vkinff} is also valid for any stationary observer not necessarily aligned with the test particle.

\subsection{Test particles with circular geodesic orbits}
\label{sec:circ}

Another important and interesting case is that of a test particle with circular geodesic orbit at radius $r_1>3m$, that we can suppose equatorial (without loss of generality) and whose 4-velocity is then given by
\begin{equation}
\label{eq:4velcirc}
U'=\left( \sqrt{\frac{r_1}{r_1-3m}},0,0,\frac{1}{r_1}\sqrt{\frac{m}{r_1-3m}}\right) .
\end{equation}
This case is very hard to study analytically and only numerical results have been obtained. Nevertheless, taking into account Theorem \ref{theo0} and the expression of $U'^t$ given in \eqref{eq:4velcirc}, the square of the modulus of the kinematic relative velocity of the test particle with respect to any stationary observer is constant and given by
\begin{equation}
\label{vkincirc}
\| V_{\mathrm{kin}}\| ^2=\frac{m}{r_1-2m}.
\end{equation}
Since \eqref{vkincirc} only depends on $r_1$, the behavior of the kinematic relative velocity is even more ``Newtonian'' compared with the other three relative velocities whose moduli are not constant (see Figure \ref{velschw}).

\begin{figure}[tbp]
\begin{center}
\includegraphics[width=0.9\textwidth]{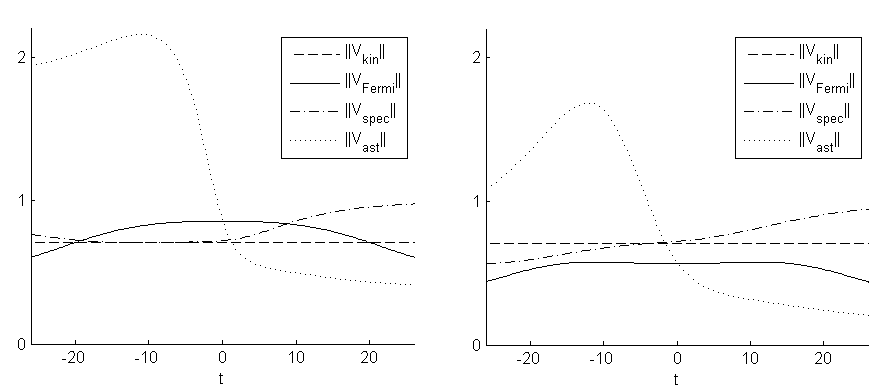}
\end{center}
\caption{Retarded comparison (see \cite{BK12}) of the moduli of the kinematic, Fermi, spectroscopic and astrometric relative velocities of a test particle with equatorial circular geodesic orbit with radius $r_1=4$, $\theta _1=\pi /2$ and $\varphi _1=0$ at $t=0$, with respect to a stationary observer at $r_0=3$ (left) and $r_0=8$ (right), $\theta _0=\pi /2$ and $\varphi _0=0$, in the Schwarzschild metric with $m=1$. The modulus of the kinematic relative velocity (dashed line) remains constant and equal to $\sqrt{1/2}$. They have been numerically computed with a relative error less than $10^{-6}$.} \label{velschw}
\end{figure}

\section{Final remarks}
\label{sec5}

Theorem \ref{theo0} assures that the modulus of the kinematic relative velocity of a test particle with respect to a stationary observer can be computed choosing any stationary observer.
Moreover, stationary observers are \textit{kinematically comoving} (see \cite{Bol07}) between them as it is proved in Section \ref{sec:stp}. In fact, from Theorem \ref{theo1} we obtain the next result that holds in any metric, generalizing Theorem \ref{theo0}: given a congruence of kinematically comoving observers, the modulus of the kinematic relative velocity of a test particle with respect to any observer of the congruence remains constant.

\begin{theorem}
\label{theo1}
Let $U$ be a congruence of kinematically comoving observers. Given a test particle with 4-velocity $u'$ at $q$ and given $p$ such that $q\in L_{p,u}$ (with $u:=U_p$), we have
\begin{equation}
\label{eqtheo1}
\left\Vert v_{\mathrm{kin}}\right\Vert = \left\Vert v\right\Vert ,
\end{equation}
where $v_{\textrm{kin}}$ is the kinematic relative velocity of $u'$ with respect to $u$, and $v$ is the usual relative velocity of $u'$ with respect to $U_q$.
\end{theorem}

\begin{proof}
Since $U$ is a congruence of kinematically comoving observers, the kinematic relative velocity of $U_q$ with respect to $u$ is zero, and so, by \eqref{vkinmod} we have that $g\left( U_q,\tau _{pq}u\right) ^2=1$. Hence, $g\left( U_q,\tau _{pq}u\right) =-1$ because they are timelike and future-pointing, and then $\tau _{pq}u=U_q$, because they are also unit. Finally, by \eqref{vusual} and \eqref{vkinmod} we have
$
\left\Vert v_{\mathrm{kin}}\right\Vert ^2 = 1-\frac{1}{g\left( u',U_q\right) ^2} = \left\Vert v\right\Vert ^2
$
and so, \eqref{eqtheo1} holds.
\end{proof}

Theorem \ref{theo1} can be also expanded to spectroscopic relative velocity taking into account a congruence of \textit{spectroscopically comoving} observers and the past-pointing horismos submanifold $E^-_p$ (see \cite{Beem81,Bol07}) instead of $L_{p,u}$; but stationary observers are not spectroscopically comoving between them in Schwarzschild spacetime, and so, we can not apply it in our case. On the other hand, this result does not hold in general for Fermi or astrometric relative velocities: for example, the congruence of stationary observers is also \textit{Fermi comoving} and \textit{astrometrically comoving} (see Section \ref{sec:stp}), but the modulus of the Fermi or the astrometric relative velocity of a test particle depends on the chosen stationary observer (see Figure \ref{velschw}).

All these facts and the results obtained in Section \ref{sec4} (specially in Section \ref{sec:circ}, where it is shown that the modulus of the kinematic relative velocity of a test particle with circular geodesic orbit with respect to a stationary observer remains constant) leads to the conclusion that the kinematic relative velocity can be interpreted as the most ``Newtonian-like'' velocity of the four geometric velocities introduced in \cite{Bol07}.



\end{document}